\documentclass[nopacs,11pt]{revtex4}

\usepackage[titletoc,title]{appendix}
\usepackage{graphicx,epic,eepic,epsfig,latexsym,amssymb,verbatim,color}
\usepackage{dsfont}
\usepackage{float}
 \usepackage{tikz}
 \usepackage{mathtools}
 \usepackage{multirow}
 \usepackage{pdfpages}
\usepackage{amsmath}

\usepackage{color}
\definecolor{darkred}{rgb}{0.8,0.1,0.1}
\definecolor{darkblue}{rgb}{0.0,0.0,0.55}

\usepackage{hyperref}
\hypersetup{colorlinks=true,citecolor=blue,linkcolor=blue,filecolor=blue,urlcolor=blue,breaklinks=true}

\usepackage[marginal]{footmisc}
\usepackage{url}
\usepackage{theorem}
\newtheorem{definition}{Definition}
\newtheorem{proposition}[definition]{Proposition}

\newtheorem{theorem}[definition]{Theorem}
\newtheorem{corollary}[definition]{Corollary}

\def\squareforqed{\hbox{\rlap{$\sqcap$}$\sqcup$}}
\def\qed{\ifmmode\squareforqed\else{\unskip\nobreak\hfil
\penalty50\hskip1em\null\nobreak\hfil\squareforqed
\parfillskip=0pt\finalhyphendemerits=0\endgraf}\fi}
\def\endenv{\ifmmode\;\else{\unskip\nobreak\hfil
\penalty50\hskip1em\null\nobreak\hfil\;
\parfillskip=0pt\finalhyphendemerits=0\endgraf}\fi}
\newenvironment{proof}{\noindent \textbf{{Proof~} }}{\qed}
\newenvironment{remark}{\noindent \textbf{{Remark~}}}{}

\mathchardef\ordinarycolon\mathcode`\:
\mathcode`\:=\string"8000
\def\vcentcolon{\mathrel{\mathop\ordinarycolon}}
\begingroup \catcode`\:=\active
  \lowercase{\endgroup
  \let :\vcentcolon
  }

\makeatletter
\def\resetMathstrut@{%
    \setbox\z@\hbox{%
        \mathchardef\@tempa\mathcode`\[\relax
        \def\@tempb##1"##2##3{\the\textfont"##3\char"}%
        \expandafter\@tempb\meaning\@tempa \relax
    }%
    \ht\Mathstrutbox@\ht\z@ \dp\Mathstrutbox@\dp\z@}
\makeatother
\begingroup
\catcode`(\active \xdef({\left\string(}
\catcode`)\active \xdef){\right\string)}
\endgroup
\mathcode`(="8000 \mathcode`)="8000

\usepackage{pifont}
\newcommand{\nc}{\newcommand}
\nc{\rnc}{\renewcommand}
\nc{\beg}{\begin{equation}}
\nc{\eeq}{{\end{equation}}}
\nc{\beqa}{\begin{eqnarray}}
\nc{\eeqa}{\end{eqnarray}}
\nc{\lbar}[1]{\overline{#1}}
\nc{\bra}[1]{\langle#1|}
\nc{\ket}[1]{|#1\rangle}
\nc{\ketbra}[2]{|#1\rangle\!\langle#2|}
\nc{\braket}[2]{\langle#1|#2\rangle}

\nc{\proj}[1]{| #1\rangle\!\langle #1 |}
\nc{\avg}[1]{\langle#1\rangle}
\nc{\Rank}{\operatorname{Rank}}
\nc{\smfrac}[2]{\mbox{$\frac{#1}{#2}$}}
\nc{\tr}{\operatorname{Tr}}
\nc{\ox}{\otimes}
\nc{\dg}{\dagger}
\nc{\dn}{\downarrow}
\nc{\cA}{{\cal A}}
\nc{\cB}{{\cal B}}
\nc{\cC}{{\cal C}}
\nc{\cD}{{\cal D}}
\nc{\cE}{{\cal E}}
\nc{\cF}{{\cal F}}
\nc{\cG}{{\cal G}}
\nc{\cH}{{\cal H}}
\nc{\cI}{{\cal I}}
\nc{\cJ}{{\cal J}}
\nc{\cK}{{\cal K}}
\nc{\cL}{{\cal L}}
\nc{\cM}{{\cal M}}
\nc{\cN}{{\cal N}}
\nc{\cO}{{\cal O}}
\nc{\cP}{{\cal P}}
\nc{\cQ}{{\cal Q}}
\nc{\cR}{{\cal R}}
\nc{\cS}{{\cal S}}
\nc{\cT}{{\cal T}}
\nc{\cV}{{\cal V}}
\nc{\cX}{{\cal X}}
\nc{\cY}{{\cal Y}}
\nc{\cZ}{{\cal Z}}
\nc{\cW}{{\cal W}}
\nc{\csupp}{{\operatorname{csupp}}}
\nc{\qsupp}{{\operatorname{qsupp}}}
\nc{\var}{{\operatorname{var}}}
\nc{\rar}{\rightarrow}
\nc{\lrar}{\longrightarrow}
\nc{\polylog}{{\operatorname{polylog}}}
\nc{\1}{{\mathds{1}}}
\nc{\wt}{{\operatorname{wt}}}
\nc{\av}[1]{{\left\langle {#1} \right\rangle}}
\nc{\supp}{{\operatorname{supp}}}

\def\di{\diamondsuit}

\def\ve{\varepsilon}

\def\o{\omega}

\def\G{\Gamma}

\def\O{\Omega}

\nc{\RR}{{{\mathbb R}}}
\nc{\CC}{{{\mathbb C}}}
\nc{\FF}{{{\mathbb F}}}
\nc{\NN}{{{\mathbb N}}}
\nc{\ZZ}{{{\mathbb Z}}}
\nc{\PP}{{{\mathbb P}}}
\nc{\QQ}{{{\mathbb Q}}}
\nc{\UU}{{{\mathbb U}}}
\nc{\EE}{{{\mathbb E}}}
\nc{\id}{{\operatorname{id}}}

\nc{\CHSH}{{\operatorname{CHSH}}}

\nc{\be}{\begin{equation}}
\nc{\ee}{{\end{equation}}}
\nc{\bea}{\begin{eqnarray}}
\nc{\eea}{\end{eqnarray}}
\nc{\<}{\langle}
\rnc{\>}{\rangle}
\nc{\Hom}[2]{\mbox{Hom}(\CC^{#1},\CC^{#2})}
\nc{\rU}{\mbox{U}}

\nc{\ob}[1]{#1}

\nc{\SEP}{{\text{SEP}}}
\nc{\NS}{{\text{NS}}}
\nc{\LOCC}{{\text{LOCC}}}
\nc{\PPT}{{\text{PPT}}}
\nc{\EXT}{{\text{EXT}}}
\nc{\Sym}{{\operatorname{Sym}}}

\nc{\ERLO}{{E_{\text{r,LO}}}}
\nc{\ERLOCC}{{E_{\text{r,LOCC}}}}
\nc{\ERPPT}{{E_{\text{r,PPT}}}}
\nc{\ERLOCCinfty}{{E^{\infty}_{\text{r,LOCC}}}}
\nc{\Aram}{{\operatorname{\sf A}}}

\newcommand{\eps}{\varepsilon}

\begin{document}
\title{Semidefinite programming converse bounds for quantum communication}
 \author{Xin Wang$^{1,2}$}
 \email{xwang93@umd.edu}
\author{Kun Fang$^{1}$}
 \email{kun.fang-1@student.uts.edu.au}
 \author{Runyao Duan$^{3,1}$}
 \email{runyao.duan@uts.edu.au}
 \affiliation{$^1$Centre for Quantum Software and Information, Faculty of Engineering and Information Technology, University of Technology Sydney, NSW 2007, Australia}
 \affiliation{$^2$Joint Center for Quantum Information and Computer Science, University of Maryland, College Park, Maryland 20742, USA}
\affiliation{$^3$
Institute for Quantum Computing, Baidu Inc., Beijing 100193, China}

\thanks{
This work is an extended version of the previous work in~\cite{Wang2016a}.}

\begin{abstract}
We derive several efficiently computable converse bounds for quantum communication over quantum channels in both the one-shot and asymptotic regime.  First, we derive one-shot semidefinite programming (SDP) converse bounds on the amount of quantum information that can be transmitted over a single use of a quantum channel, which improve the previous bound from [Tomamichel/Berta/Renes, \emph{Nat. Commun.} \textbf{7}, 2016]. As applications, we study quantum communication over depolarizing channels and amplitude damping channels with
finite resources. Second, we find an SDP strong converse bound for the quantum capacity of an arbitrary quantum channel, which means the fidelity of any sequence of codes with a rate exceeding this bound will vanish exponentially fast as the number of channel uses increases. Furthermore, we prove that the SDP strong converse bound improves the \emph{partial transposition bound} introduced by Holevo and Werner. Third, we prove that this SDP strong converse bound is equal to the so-called \emph{max-Rains information}, which is an analog to the Rains information introduced in [Tomamichel/Wilde/Winter, \emph{IEEE Trans. Inf. Theory} \textbf{63}:715, 2017]. Our SDP strong converse bound is weaker than the Rains information, but it is efficiently computable for general quantum channels.
\end{abstract}

\maketitle
\section{Introduction}
\subsection{Background}
The reliable transmission of quantum information via noisy quantum channels is a fundamental problem in quantum information theory.
The quantum capacity of a noisy quantum channel is the optimal rate at which it can convey quantum bits (qubits) reliably over asymptotically many uses of the channel.  The theorem by Lloyd, Shor, and Devetak (LSD)~\cite{Lloyd1997,Shor2002a,Devetak2005a} and the work in Refs.~\cite{Schumacher1996a,Barnum2000,Barnum1998} show
that the quantum capacity is equal to the regularized coherent information.  In general, the regularization of coherent information is necessary since the coherent information can be superadditive.
The quantum capacity is notoriously difficult to evaluate since it is characterized by a multi-letter, regularized expression and it is not even known to be computable~\cite{Cubitt2015,Elkouss2015}. Even for the qubit depolarizing channel, the quantum capacity is still unsolved despite substantial effort in the past two decades (see e.g.,~\cite{DiVincenzo1998a,Fern2008,Smith2007,Smith2008a,Sutter2014,Leditzky2017,Leung2015a}). Our understanding of quantum capacity is quite limited, and we even do not know the exact threshold value of the depolarizing noise where the capacity goes to zero.

The converse part of the LSD theorem states that if the rate exceeds the quantum capacity, then the fidelity of any coding scheme cannot approach one in the limit of many channel uses. A strong converse property leaves no room for the trade-off between rate and error, i.e.,  the error probability vanishes in the limit of many channel uses whenever the rate exceeds the capacity. For classical channels,  Wolfowitz~\cite{Wolfowitz1978} established the strong converse property for the classical capacity.
For quantum channels, the strong converse property for the classical capacity was confirmed for several classes of channels~\cite{Ogawa1999,Winter1999,Koenig2009,Wilde2013a,Wilde2014a,Wang2016g}. 

For quantum communication, the strong converse property was studied in Ref.~\cite{Tomamichel2015a} and the strong converse of generalized dephasing channels was established~\cite{Tomamichel2015a}. Given an arbitrary quantum channel, a previously known efficiently computable strong converse bound on the quantum capacity for general channels is the partial transposition bound~\cite{Holevo2001,Muller-Hermes2015}. Recently, the Rains information~\cite{Tomamichel2015a} was established to be a strong converse bound for quantum communication. 
There are other known upper bounds for quantum capacity~\cite{Smith2008a,Sutter2014,Gao2015a,Bruß1998, Cerf2000, Wolf2007,Smith2008b,Leditzky2017} and most of them require specific settings to be computable and relatively tight. 

Moreover, in a practical setting, the number of quantum channel uses is finite, and one has to make a trade-off between the transmission rate and error tolerance. For both practical and theoretical interest, it is important to optimize the trade-off for the rate and infidelity of quantum communication with finite resources.  The study of this finite blocklength setting has recently attracted great interest in classical information theory (e.g.,~\cite{Hayashi2009,Polyanskiy2010}) as well as in quantum information theory (e.g.,~\cite{Tomamichel2013a,Wang2012,Leung2015c,Matthews2014,Tomamichel2015b,Beigi2015,Tomamichel2016}). 

\subsection{Summary of results}
In this paper, we focus on quantum communication via noisy quantum channels in both the one-shot and asymptotic settings.  We study the quantum capacity assisted with positive partial transpose preserving (PPT) and no-signalling (NS) codes ~\cite{Leung2015c}. The PPT codes include all the operations that can be implemented by local operations and classical communication while the NS codes are potentially stronger than entanglement-assisted codes.

In section \ref{Converse bounds for non-asymptotic quantum communication}, we consider the non-asymptotic quantum capacity. We first introduce the one-shot $\ve$-infidelity quantum capacity with PPT-assisted (and NS-assisted) codes and characterize it as an optimization problem. Based on this optimization, we provide semidefinite programming (SDP) bounds to evaluate the one-shot capacity with a given infidelity tolerance. Compared with the previous efficiently computable converse bound given in Ref.~\cite{Tomamichel2016}, we show that our SDP converse bounds are tighter in general and can be strictly tighter for basic channels such as the qubit amplitude damping channel and the qubit depolarizing channel.

In section \ref{Strong converse bound for quantum communication}, we investigate quantum communication via quantum channels in the asymptotic setting. We first present an SDP strong converse bound, denoted as $Q_\G$, on the quantum capacity for a general quantum channel. This bound has some nice properties, such as additivity with respect to tensor products of quantum channels. In particular, $Q_\G$ is a channel analog of the SDP entanglement measure introduced in Ref.~\cite{Wang2016}, and we show here that it
is equal to the so-called max-Rains information. This result implies that $Q_\G$ is no better, in general, as an upper bound on quantum capacity than the Rains information~\cite{Tomamichel2015a}. However, $Q_\G$ is efficiently computable for general quantum channels. Finally, we show that $Q_\G$ improves the partial transposition bound \cite{Holevo2001}.

\section{Preliminaries}\label{sec:pre}
In the following, we will frequently use symbols such as $A$ (or $A'$) and $B$ (or $B'$) to denote (finite-dimensional) Hilbert spaces associated with Alice and Bob, respectively. We use $d_A$ to denote the dimension of system $A$. The set of linear operators acting on $A$ is denoted by $\cL(A)$. The set of positive operators acting on $A$ is denoted by $\cP(A)$. The set of positive operators with unit trace is denoted by $\cS(A)$, while the set of positive operators with trace no greater than $1$ is denoted by $\cS_\leq(A)$. We usually write an operator with a subscript indicating the system that the operator acts on, such as $M_{AB}$, and write $M_A:=\tr_B M_{AB}$. Note that for a linear operator $X\in\cL(A)$, we define $|X|=\sqrt{X^\dagger X}$, where $X^\dagger$ is the adjoint operator of $X$, and the trace norm of $X$ is given by $\|X\|_1=\tr |X|$. 
A quantum channel $\cN_{A'\to B}$ is simply a completely positive (CP) and trace-preserving (TP) linear map from $\cL(A')$ to $\cL(B)$.  The Choi-Jamio\l{}kowski matrix of $\cN$ is given by $J_{\cN}=\sum_{ij} \ketbra{i_A}{j_{A}} \ox \cN(\ketbra{i_{A'}}{j_{A'}})$, where $\{\ket{i_A}\}$ and $\{\ket{i_{A'}}\}$ are orthonormal bases on isomorphic Hilbert spaces $A$ and $A'$, respectively. 

A positive semidefinite (PSD) operator  $E \in \cL(A \ox B)$ is
said to be a positive
partial transpose operator (PPT) if $E^{T_{B}}\geq 0$, where ${T_{B}}$ means the partial transpose with respect to the party
$B$, i.e., $(\ketbra{ij}{kl})^{T_{B}}=\ketbra{il}{kj}$.
A bipartite operation $\Pi_{A_i B_i\to A_o B_o}$ is PPT if and only if its Choi-Jamio\l{}kowski matrix  is PPT~\cite{Rains2001}. 
The set of PPT operations include all operations that can be implemented by local operations and classical communication. 

Semidefinite programming~\cite{Vandenberghe1996} is a useful tool in the study of quantum information and computation with many applications.  
In this work, we use the CVX software
\cite{Grant2008} and QETLAB (A Matlab Toolbox for Quantum
Entanglement)~\cite{NathanielJohnston2016}
to solve SDPs.

\section{Converse bounds for non-asymptotic quantum communication}
\label{Converse bounds for non-asymptotic quantum communication}
\subsection{One-shot $\ve$-error capacity and finite resource trade-off}

In this section, we are interested in quantum communication via noisy channels with finite resources. Suppose Alice shares a maximally entangled state $\Phi_{A_iR}$ with a reference system $R$ to which she has no access.  The goal is to design a quantum coding protocol such that Alice can transfer her share of this maximally entangled state to Bob with very high fidelity. To this end, Alice first performs an encoding operation $\cE_{A_i\to A_o}$ on system $A_i$ and then transmits the prepared state through the channel $\cN_{A_o\to B_i}$. The resulting state turns out to be $\cN_{A_o\to B_i}\circ \cE_{A_i\to A_o}(\Phi_{A_iR})$. 
After Bob receives the state, he performs a decoding operation $\cD_{B_i\to B_o}$ on system $B_i$, where $B_o$ is some system of the same dimension as $A_i$. The final resulting state will be $\rho_{final}=\cD_{B_i\to B_o} \circ \cN_{A_o\to B_i}\circ \cE_{A_i\to A_o}(\Phi_{A_iR})$. The target of quantum coding is to optimize the fidelity between $\rho_{final}$ and the maximally entangled state $\Phi_{A_iR}$.

One could further imagine the coding protocol as a general super-operator $\Pi_{A_iB_i\to A_oB_o}$. The authors of Ref.~\cite{Chiribella2008} showed that a two-input and two-output CPTP map $\Pi_{A_iB_i\to A_oB_o}$ sends any CPTP map $\cN_{A_o\to B_i}$ to another CPTP map $\cM_{A_i\to B_o}$ if and only if $\Pi_{A_iB_i\to A_oB_o}$ is {\rm B} to {\rm A} no-signalling (see also \cite{Duan2016}). Such bipartite operation $\Pi$ is called deterministic super-operator or semi-causal quantum operation.
Let $\cM_{A_i\to B_o}$ denote the resulting composition channel of a deterministic super-operator $\Pi_{A_iB_i\to A_oB_o}$ and a channel $\cN_{A_o\to B_i}$. We write $\cM=\Pi\circ\cN$ for simplicity. Then there exist CPTP maps $\cE_{A_i\to A_oC}$ and $\cD_{B_iC\to B_o}$, where $C$ is a quantum register, such that \cite{Eggeling2002a,Chiribella2008,Duan2016}
\begin{equation}
\cM_{A_i\to B_o} = \cD_{B_iC\to B_0} \circ \cN_{A_o \to B_i} \circ \cE_{A_i\to A_oC}.
\end{equation}

The \textit{no-signalling} ({\rm NS}) codes \cite{Leung2015c,Cubitt2011,Duan2016} correspond to the bipartite quantum operations which are  no-signalling from {\rm B} to {\rm A} and vice-versa. The {\rm PPT} codes \cite{Leung2015c} correspond to the deterministic super-operators which are also PPT.
The \textit{non-signalling and PPT-preserving} ({\rm NS$\cap$PPT}) codes correspond to the quantum no-signalling operations which are also PPT. Moreover, a bipartite quantum operation $\Pi_{A_iB_i\to A_oB_o}$ is called unassisted code (UA) if it can be represented as $\Pi_{A_iB_i\to A_oB_o} = \cE_{A_i\to A_o}\circ \cD_{B_i\to B_o}$. In the following, $\O$ denotes specific classes of codes, i.e., $\O \in \{\rm{UA, NS\cap PPT, PPT}\}$. 

\begin{figure}
\centering
\begin{minipage}{.38\textwidth}
\begin{tikzpicture}[scale = 0.96]
\def\xa{1};\def\xb{2};\def\xc{3};\def\xd{4};\def\xe{5};\def\xf{6};\def\xg{7};
\def\ya{0.8};\def\yc{-1.6};\def\yd{-2.4};  \def\o{0.2} \def\ye{1}; \def\yf{2};
\pgfmathsetmacro\yb{-\ya}
\draw[thick,-] (0,\ye) -- node[above,shift={(-0.2,0)}] {$R$} (\xa,\yf);
\draw[thick,->] (\xa,\yf) -- node[above] {} (\xg,\yf);
\draw[thick,->] (0,\ye) -- node[below,xshift = -0.2cm,yshift = 0.1cm] {$A_i$} (\xa,0);
\draw[thick,->] (\xb,0) -- node[above] {$A_o$} (\xc,0);
\draw[thick,->] (\xd,0) -- node[above] {$B_i$} (\xe,0);
\draw[thick,->] (\xf,0) -- node[above] {$B_o$} (\xg,0);
\draw[thick,->] (\xb,\yc) -- node[below] {$C$} (\xe,\yc);
\draw (\xa,\ya) rectangle (\xb,\yd) node[midway] {\Large $\cE$};
\draw (\xc,\ya) rectangle (\xd,\yb) node[midway] {\Large $\cN$};
\draw (\xe,\ya) rectangle (\xf,\yd) node[midway] { \Large $\cD$};
\draw[dashed,blue,fill=black!20,opacity=0.3] (\xa-\o,\ya+\o) -- (\xb+\o,\ya+\o) -- (\xb+\o,\yc+\o) -- (\xe-\o,\yc+\o)-- (\xe-\o,\ya+\o)  -- (\xf+\o,\ya+\o) -- (\xf+\o,\yd-\o) -- 
(\xa-\o,\yd-\o) -- (\xa-\o,\ya+\o);
\node[black,left,shift={(0.1,0.2)}] (d) at (\xa-\o-0.1,\yd-\o+0.1) {\Large $\Pi$};
\end{tikzpicture}
\label{fig:test2}
\end{minipage}
\hspace{1.3cm}
\begin{minipage}{.38\textwidth}
\vspace{-0.5cm}
\begin{tikzpicture}
\def\xbb{0.6};\def\xb{1.5};\def\xsh{1.3};\def\ysh{1};
\def\ya{1.3};\def\yc{3.4};\def\lo{0.2};\def\loo{0.4};
\pgfmathsetmacro\xa{-\xb};\pgfmathsetmacro\xaa{-\xbb};
\pgfmathsetmacro\xc{\xa-\xsh/2};\pgfmathsetmacro\xd{\xa+\xsh/2};
\pgfmathsetmacro\xe{\xb-\xsh/2};\pgfmathsetmacro\xf{\xb+\xsh/2};
\pgfmathsetmacro\yb{\ya+\ysh};
\pgfmathsetmacro\yab{(\ya+\yb)/2};

\draw[thick,->] (\xc-1,\yb+0.4) -- node[left,shift={(0.1,-0.2)}] 
{$A_i$} (\xc,\yab);
\draw[thick,-] (\xc-1,\yb+0.4) -- node[above,shift={(-0.2,0)}] 
{$R$} (\xc,\yab+1.8);
\draw[thick,->] (\xc,\yab+1.8) -- node[above,shift={(-0.2,0)}] 
{} (\xf+1.1,\yab+1.8);
\draw[thick,->] (\xf,\yab) -- node[above,shift={(0.2,0)}] {$B_o$} (\xf+1.1,\yab);
\draw (\xc,\yb) rectangle (\xd,\ya) node[midway] {\large $\cE$};
\draw (\xe,\yb) rectangle (\xf,\ya) node[midway] {\large $\cD$};
\draw[thick,->] (\xa,\ya) -- (\xa,0) -- node[below] {$A_o$} (\xaa,0);
\draw[thick,<-] (\xb,\ya) -- (\xb,0) -- node[below] {$B_i$} (\xbb,0);
\draw (\xaa,\ysh/2) rectangle (\xbb,-\ysh/2) node[midway] {\large $\cN$};
\draw[dashed,blue,fill=black!20,opacity=0.3] (\xc-\lo,\yb+\lo) rectangle (\xf+\lo,\ya-0.6);
\node[black] (dots) at (0,\yb-0.7) {\Large $\Pi$};
\draw[thick,dotted,black!70] (\xc-\loo,\yb+\loo) rectangle (\xf+\loo,-\ysh/2-\lo);
\node[black!50] (dots) at (\xc-0.15,-\ysh/2) {$\cI$};
\end{tikzpicture}
\end{minipage}
\caption{A deterministic super-operator $\Pi_{A_iB_i\to A_oB_o}$ is equivalently the coding scheme ($\cE$,$\cD$) with free extra resources such as entanglement. The whole operation aims to simulate a noiseless quantum channel $\cI_{A_i\to B_o}$ using a given noisy quantum channel $\cN_{A_o\to B_i}$ and the bipartite code $\Pi$.}
\label{fig:QNSC}
\end{figure}
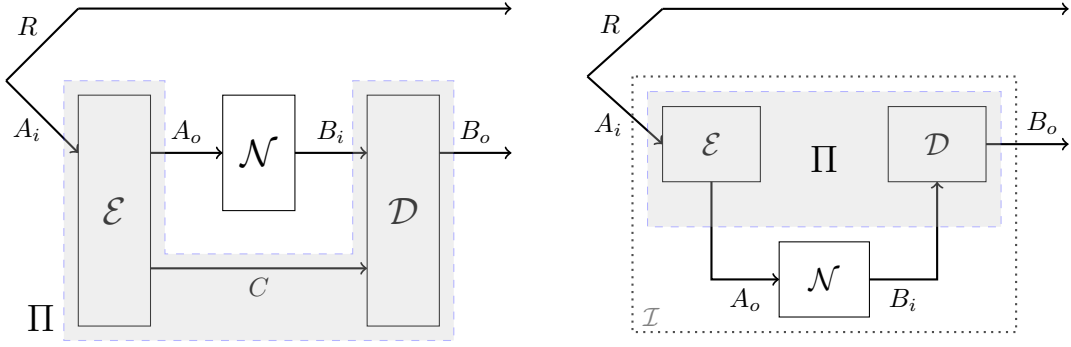

\begin{definition}
The maximum channel fidelity of $\cN$ assisted by the $\O$ code is defined by 
\begin{equation}
F_\O(\cN,k) := \sup_{\Pi} \tr (\Phi_{B_oR}\cdot \Pi_{A_iB_i\to A_oB_o}\circ \cN_{A_o\to B_i} (\Phi_{A_iR})),
\end{equation}
where $\Phi_{A_iR}$ and $\Phi_{B_oR}$ are maximally entangled states, $k = dim |A_i| = dim |B_o|$ is called code size and the supremum is taken over the $\O$ codes ($\O \in \{\rm{UA, NS\cap PPT, PPT}\}$).
\end{definition}

\begin{definition}
For a given quantum channel $\cN$ and error tolerance $\ve$, the one-shot $\ve$-error quantum capacity assisted by $\O$ codes is defined by 
\begin{equation}
\label{oneshot capacity}
Q_{\O}^{(1)}(\cN,\ve) := \log \max \left\{k \in \mathbb{N}: F_\O(\cN,k) \geq 1-\ve \right\},\end{equation}
where $\O \in \{\rm{UA, NS\cap PPT, PPT}\}$. In the following, we write $Q_{\rm UA}^{(1)}(\cN,\ve)= Q^{(1)}(\cN,\ve)$ for simplicity.

The corresponding asymptotic quantum capacity is then given by \begin{equation}
\label{asymptotic capacity}
Q_{\O}(\cN) = \lim_{\ve \rightarrow 0}\lim_{n\rightarrow \infty} \frac1n Q_{\O}^{(1)}(\cN^{\ox n},\ve).
\end{equation}
\end{definition}

The authors of Ref.~\cite{Leung2015c} showed that
the maximum channel fidelity assisted with ${\rm NS\cap PPT}$ codes is given by the following SDP:
\begin{equation}\label{PPT prime}
\begin{split}
F_{\rm NS\cap PPT}(\cN, k)= \max & \tr J_{\cN}W_{AB}\\
\text{ s.t.}& \  0 \leq W_{AB} \leq \rho_A \ox \1_B, \tr \rho_A=1,\\
& - k^{-1} \rho_A \ox \1_B \le W_{AB}^{T_{B}} \le k^{-1}\rho_A \ox \1_B,\\&
\tr_A W_{AB} = k^{-2}\1_B \ (\rm{NS}).
\end{split}\end{equation}
To obtain $F_{\rm PPT}(\cN, k)$, one only needs to remove the NS constraint.

Combining Eqs.~\eqref{oneshot capacity} and \eqref{PPT prime}, one can derive the following proposition. It is worth noting that Eq.~(\ref{tradeoff sdp primal}) is not an SDP in general, due to the non-linear term $m \rho_A$ and the condition $\tr_A W_{AB}=m^2\1_B$. But in next subsection, we will derive  several semidefinite relaxations of this optimization problem.
\begin{proposition}
\label{optimization characterization}
For any quantum channel $\cN_{A'\rightarrow B}$ with Choi-Jamio\l{}kowski matrix $J_{\cN} \in \cL(A\ox B)$ and given error tolerance $\ve$, its one-shot $\ve$-error quantum capacity assisted with PPT codes can be simplified as the following optimization problem:
\begin{equation}\label{tradeoff sdp primal}
\begin{split}
Q_{\rm PPT}^{(1)}(\cN,\ve)  =  -\log \min & \ m \\
\text{ s.t. } &
\tr J_{\cN} W_{AB} \geq 1-\varepsilon,
0 \leq W_{AB} \leq \rho_A\otimes \1_B,\\
& \tr \rho_A = 1, -m\rho_A \otimes \1_B \leq W_{AB}^{T_B} \leq m\rho_A \otimes \1_B.
\end{split}
\end{equation}
If the codes are also non-signalling, we can have the same optimization for $Q_{NS \cap PPT}^{(1)}(\cN,\ve)$ with the additional constraint  $\tr_A W_{AB}=m^2\1_B$.
\end{proposition}

%

\subsection{Improved SDP converse bounds for quantum communication}
 
To better evaluate the quantum communication rate with finite resources, we introduce several SDP converse bounds for quantum communication with the assistance of PPT (and NS) codes.  In Theorem \ref{improved sdp bound theorem}, we further prove that our SDP bounds are tighter than the one introduced in Ref.~\cite{Tomamichel2016}.

\vspace{0.1cm}
Specifically, the authors of Ref.~\cite{Tomamichel2016} established that $-\log f(\cN, \ve)$ is a converse bound on one-shot $\ve$-error quantum capacity, i.e., $Q^{(1)}(\cN,\ve) \leq -\log f(\cN, \ve)$ where 
\begin{equation}\label{SDP Marco}
\begin{split}
f(\cN,\eps) = \min & \tr S_A\\
\text{s.t.}& \tr W_{AB}J_{\cN}\ge 1-\eps, S_A, \Theta_{AB}\geq 0, \tr \rho_A=1,\\
& 0\le W_{AB}\leq\rho_A \ox \1_B, S_A\otimes \1_B\geq W_{AB}+\Theta_{AB}^{T_B}.
\end{split} 
\end{equation}

Here, we introduce a hierarchy of SDP converse bounds on the one-shot $\ve$-error capacity based on the optimization problem in Eq.~\eqref{tradeoff sdp primal}. If we relax the term $m\rho_A$ to a single variable $S_A$, we will obtain $g(\cN,\ve)$, where
\begin{equation}
\label{PPT g}
\begin{split}
g(\cN,\ve) := \min & \tr S_A\\
\text{ s.t. }& \tr J_{\cN}W_{AB}\ge 1-\ve,0 \leq W_{AB} \leq \rho_A \ox \1_B,\\
& \tr \rho_A=1,-{S_A \ox \1_B} \le W_{AB}^{T_{B}} \le  {S_A \ox \1_B}.
\end{split}
\end{equation}

In particular, for the NS condition $\tr_A W_{AB}=m^2\1_B$, there are two different ways to get relaxations. The first one is to substitute it with $\tr_A W_{AB} =t \1_B$ and obtain SDP~$\widetilde g(\cN,\ve)$). The second one is to introduce a prior constant $\widehat m$ satisfying the inequality \begin{align}
Q_{\rm NS\cap PPT}^{(1)}(\cN,\ve) \leq -\log \widehat m\end{align}
and then obtain SDP~$\widehat g(\cN,\ve)$. Note that the second method can provide a tighter bound, but it requires one more step of calculation since we need to get the prior constant~$\widehat m$. Successively refining the value of $\widehat m$ will result in a tighter bound.
\begin{align}
\begin{split}\label{PPT NS g}
\widetilde g(\cN,\ve) := \min & \tr S_A \\
\text{ s.t. }& \tr J_{\cN}W_{AB}\geq1-\ve, 0 \leq W_{AB} \leq \rho_A \ox \1_B,\\
&\tr \rho_A=1, -{S_A \ox \1_B} \le W_{AB}^{T_{B}} \le  {S_A \ox \1_B},\\
& \tr_A W_{AB} =t \1_B.
\end{split}\\ 
\begin{split}\label{PPT NS hat g}
\widehat g(\cN,\ve) := \min & \tr S_A \\
\text{ s.t. }& \tr J_{\cN}W_{AB}\geq1-\ve, 0 \leq W_{AB} \leq \rho_A \ox \1_B,\\
&\tr \rho_A=1, -{S_A \ox \1_B} \le W_{AB}^{T_{B}} \le  {S_A \ox \1_B},\\
& \tr_A W_{AB} =t \1_B, t \geq \widehat m^2.
\end{split}
\end{align}

\begin{theorem}
\label{improved sdp bound theorem}
For any quantum channel $\cN$ and error tolerance $\ve$, the inequality chain holds
\begin{equation}
Q^{(1)}(\cN,\ve)\le Q_{\rm NS \cap PPT}^{(1)}(\cN,\ve) \leq -\log \widehat g(\cN,\ve) \le -\log \widetilde g(\cN,\ve) \le -\log g(\cN,\ve)\le -\log f(\cN,\ve).
\end{equation}
\end{theorem}

\begin{proof}
The first inequality is trivial. The third and fourth inequalities are easy to obtain since the minimization over a smaller feasible set gives a larger optimal value here.

For the second inequality,  suppose the optimal solution of (\ref{tradeoff sdp primal}) for~$Q_{\rm NS \cap PPT}^{(1)}(\cN,\ve)$,  is taken at $\{W_{AB}, \rho_A, m\}$. Let $S_A = m\rho_A,\ t = m^2$. Then we can verify that $\{W_{AB},\rho_A,S_A,t\}$ is a feasible solution to the SDP (\ref{PPT NS hat g}) of $\widehat g(\cN,\ve)$.  So $\widehat g(\cN,\ve) \leq \tr S_A = m$, which implies that $Q_{\rm NS \cap PPT}^{(1)}(\cN,\ve) = -\log m \le -\log \widehat g(\cN,\ve)$.

For the last inequality, we only need to show that $f(\cN,\eps)\le g(\cN,\eps)$.
Suppose the optimal solution of $g(\cN,\eps)$ is taken at $\{\rho_A, S_A, W_{AB}\}$. Let us choose 
$\Theta_{AB}=S_A\ox \1_B -W_{AB}^{T_B}.$
Since $S_A\ox \1_B \ge W_{AB}^{T_B}$, it is clear that $\Theta_{AB}\ge 0$ and $S_A\ox \1_B=W_{AB}+\Theta_{AB}^{T_B}.$ 
Thus, $\{S_A, \rho_A,  W_{AB}, \Theta_{AB} \}$ is a feasible solution to the SDP (\ref{SDP Marco}) of $ f(\cN, \ve)$ which implies $f(\cN, \ve) \le \tr S_A = g(\cN,\ve)$.
\end{proof}


\subsection{Examples: amplitude damping channel and depolarizing channel}
\label{channel example}
In this subsection, we focus on quantum coding with amplitude damping channels and depolarizing channels.  In Fig. \ref{AD}, we show that for the amplitude damping channel $\cN_{AD}$, our converse bound $-\log \widetilde g(\cN,\ve)$ and $-\log g(\cN,\ve)$ are both tighter than $-\log f(\cN,\ve)$. For the depolarizing channel~$\cN_D$, exploiting its symmetry, we further simplify our SDP converse bounds to linear programs.

\textbf{Example 1}:
For the amplitude damping channel  $\cN_{AD}=\sum_{i=0}^1 E_i\cdot E_i^\dag$ 
with $E_0=\ketbra{0}{0}+\sqrt{1-r}\ketbra{1}{1}$, $E_1=\sqrt{r}\ketbra{0}{1}$ 
($0\leq r\leq 1$), the differences among  $-\log f(\cN_{AD}^{\ox 2}, 0.01)$, $-\log g(\cN_{AD}^{\ox 2}, 0.01)$ and $-\log \widetilde g(\cN_{AD}^{\ox 2}, 0.01)$, are presented   in Fig.~\ref{AD}. When $r \in (0.082,0.094)$, 
$-\log \widetilde g(\cN_{AD}^{\ox 2}, 0.01) \leq -\log  g(\cN_{AD}^{\ox 2}, 0.01)< 1 < -\log f(\cN_{AD}^{\ox 2}, 0.01)$. It shows that we cannot transmit a single qubit within error tolerance $\ve = 0.01$ via $2$ copies of  amplitude damping channel where parameter $r \in (0.082,0.094)$. However, this result cannot be obtained via the converse bound $-\log f(\cN_{AD}^{\ox 2}, 0.01)$.

\begin{figure}[H]
\centering
\includegraphics[scale = 0.63]{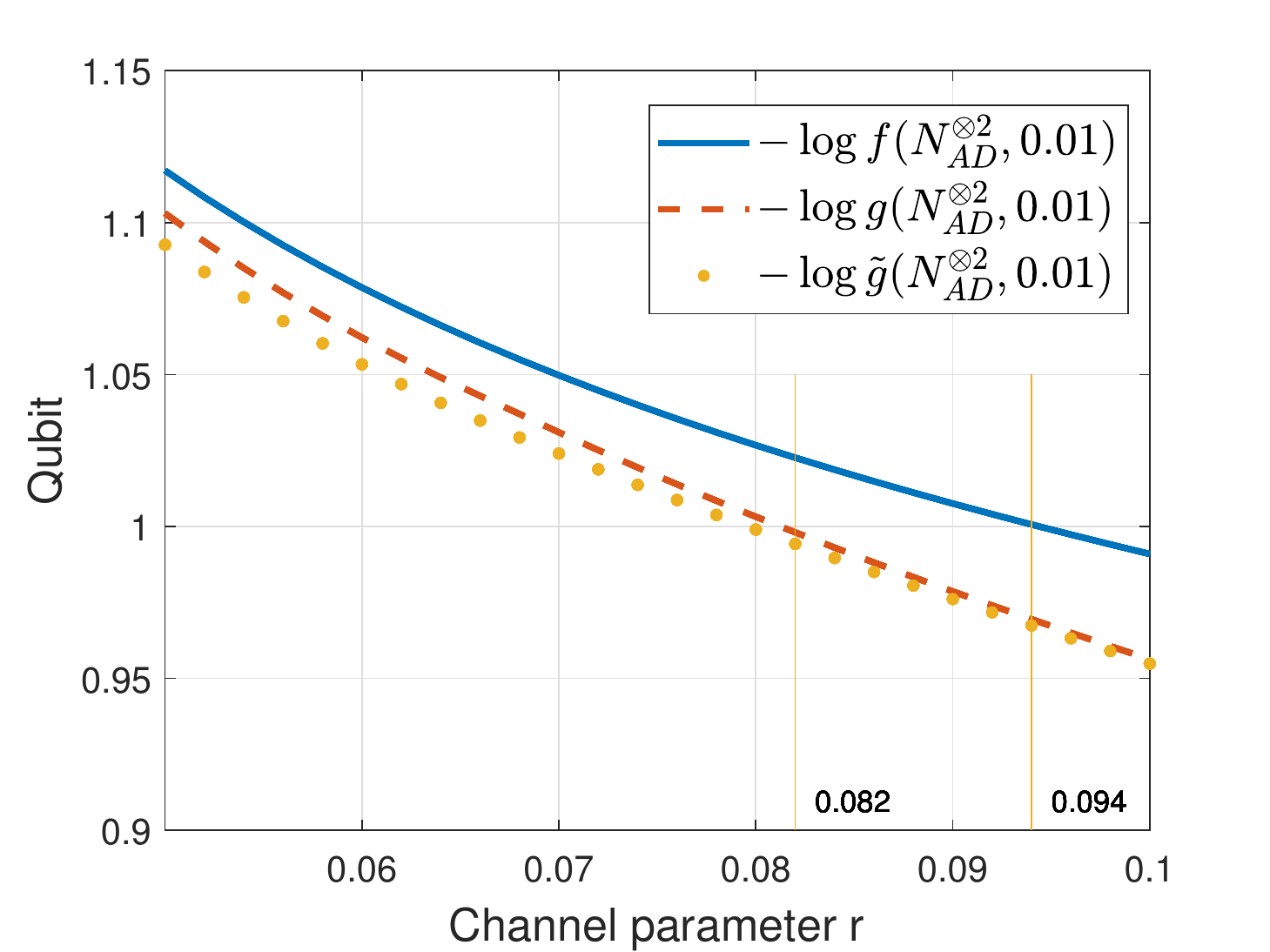}
\caption{This figure demonstrates the differences among the SDP converse bounds (i) $-\log f(\cN_{AD}^{\ox 2}, 0.01)$ (blue solid), (ii) $-\log g(\cN_{AD}^{\ox 2}, 0.01)$ (red dashed), (iii) $-\log \widetilde g(\cN_{AD}^{\ox 2}, 0.01)$ (yellow dotted), where the channel parameter $r$ ranges from $0.05$ to $0.1$.}
\label{AD}
\end{figure}

\textbf{Example 2}:
For the qubit depolarizing channel $\cN_{D}(\rho) = (1-p)\rho + \frac{p}{3}(X\rho X + Y\rho Y + Z\rho Z)$, where $X,Y,Z$ are Pauli matrices, 
the Choi matrix of $\cN_D$ is $J_{\cN} = d((1-p)\Phi +  \frac{p}{d^2-1}\Phi^\perp)$, where $d=2$, $\Phi = \frac{1}{d} \sum_{i,j=0}^{d-1} \ket{ii}\bra{jj}$ and $\Phi^\perp = \1_{AB} -\Phi$. For the $n$-fold tensor product depolarizing channel, its Choi matrix is $J_{\cN}^{\ox n} = d^n\sum_{i=0}^n f_i P_i^n(\Phi,\Phi^\perp)$, where $f_i = (1-p)^i (\frac{p}{d^2-1})^{n-i}$ and  $P_i^n(\Phi,\Phi^\perp)$ represent the sum of those $n$-fold tensor product terms with exactly $i$ copies of $\Phi$. For example,
\begin{align}
P_1^3(\Phi,\Phi^\perp) = \Phi^\perp \otimes \Phi^\perp \otimes \Phi + \Phi^\perp  \otimes \Phi \otimes \Phi^\perp + \Phi \otimes \Phi^\perp \otimes \Phi^\perp.\end{align}

Suppose $\{W_{AB},\rho_A,S_A\}$ is the optimal solution to the SDP (\ref{PPT g}) for the channel $\cN_D^{\ox n}$,
then for any local unitary $U = \bigotimes_{i=1}^n U_A^i\ox \overline{U}_B^i$, $U_A = \bigotimes_{i=1}^n U_A^i$, we know that $\{U W U^\dagger, U_A \rho_A U_A^\dagger, U_A S_A U_A^\dagger\}$ is also optimal. 
Convex combinations of optimal solutions remain optimal. Without loss of generality, we can take the optimal solution to be invariant under any local unitary $U$ and $U_A$, respectively. 
Again, since $J_{\cN}^{\ox n}$ is invariant under the symmetric group, acting by permuting the tensor factors, we can finally take the optimal solution as $W = \sum_{i=0}^n w_i P_i^n(\Phi,\Phi^\perp)$, $\rho_A = \1_A/d^n$, $S_A = s \1_A$.

Note that $P_i^n(\Phi,\Phi^\perp)$ are orthogonal projections. Thus without considering degeneracy, operator $W$ has eigenvalues $\{w_i\}_{i=0}^n$. Next, we need to know the  eigenvalues of $W^{T_B}$. Decomposing operators $\Phi^{T_B}$ and ${\Phi^{\perp}}^{T_B}$ into orthogonal projections, i.e.,
\begin{align}
\Phi^{T_B} = \frac{1}{d}(P_+ - P_-), \quad {\Phi^{\perp}}^{T_B} = (1-\frac{1}{d})P_+ + (1+\frac{1}{d})P_-
\end{align} where $P_+$ and $P_-$ are symmetric and anti-symmetric projections respectively and collecting the terms with respect to $P_k^n(P_+,P_-)$, we have
\begin{align}
&W^{T_B} = \sum_{i=0}^n w_i P_i^n(\Phi^{T_B},{\Phi^\perp}^{T_B}) =\sum_{k=0}^n (\sum_{i=0}^n x_{i,k} w_i)  P_k^n(P_+,P_-), \quad \text{where}\\
&x_{i,k} = \frac{1}{d^n}\sum_{m=\max\{0,i+k-n\}}^{\min \{i,k\}} \binom{k}{m}\binom{n-k}{i-m} (-1)^{i-m} (d-1)^{k-m}(d+1)^{n-k+m-i}.\end{align}

Since $P_k^n(P_+,P_-)$ are also orthogonal projections, $W^{T_B}$ has eigenvalues $\{t_k\}_{k=0}^n$ (without considering degeneracy), where $t_k = \sum_{i=0}^n x_{i,k} w_i$.
As for the constraint $\tr J_{\cN}^{\ox n} W_{AB}\ge 1-\varepsilon$, we have
\begin{align}
\tr  J_{\cN}^{\ox n} W  = d^n \tr \sum_{i=0}^n f_i w_i  P_i^n(\Phi, \Phi^\perp) =  d^n \sum_{i=0}^n \binom{n}{i} (1-p)^i p^{n-i} w_i  \geq 1-\varepsilon.\end{align}

Finally, substitute $\eta = sd^n$ and $m_i = w_i d^n$.  We obtain the linear program 
\begin{equation}\begin{split}\label{ED PPT LP}
g(\cN_{D}^{\ox n},\ve) = \min & \ \eta \\
\text{s.t.} & \sum_{i=0}^n \binom{n}{i} (1-p)^i p^{n-i} m_i  \geq 1-\varepsilon,\\
&0 \leq m_i \leq 1, \ i = 0,1,\cdots, n,\\
-\eta & \leq \sum_{i=0}^n x_{i,k} m_i \leq \eta, \ k = 0,1,\cdots, n.
\end{split}\end{equation}
Following a similar procedure, we have 

\begin{minipage}[t]{.4\linewidth}
\begin{equation*}
\begin{split}
f  (\cN_{D}^{\ox n},\ve)  &  = \min \ \eta\\ 
\text{ s.t. }  & \sum_{i=0}^n \binom{n}{i} (1-p)^i p^{n-i} m_i  \geq 1-\varepsilon,\\
&\ m_i + s_i \leq \eta, \ i = 0,1,\cdots, n,\\
& \eta \geq 0,\ 0 \leq m_i \leq 1, \ i = 0,1,\cdots, n\\
& \sum_{i=0}^n x_{i,k} s_i \geq 0, \ k = 0,1,\cdots, n.
\end{split}\end{equation*}
\end{minipage}
\begin{minipage}[t]{.6\linewidth}
\begin{equation*}
\begin{split}
\widehat g  (\cN_{D}^{\ox n}, \ve) & = \min \ \eta\\ 
\text{ s.t. }&\ \sum_{i=0}^n \binom{n}{i} (1-p)^i p^{n-i} m_i  \geq 1-\varepsilon,\\
& 0 \leq m_i \leq 1, \ i = 0,1,\cdots, n,\\
& -\eta \leq \sum_{i=0}^n x_{i,k} m_i \leq \eta, \ k = 0,1,\cdots, n,\\
& \frac{1}{d^{2n}}\sum_{i=0}^n \binom{n}{i} (d^2-1)^{n-i} m_i \geq \widehat m^2.
\end{split}\end{equation*}
\end{minipage}

\vspace{0.5cm}
Since $-\log \widehat g (\cN_{D}^{\ox n}, \ve)$ is a converse bound for any $\widehat m \leq 2^{-Q^{(1)}_{\text{PPT} \cap \text{NS}}(\cN_{D}^{\ox n}, \ve)}$, we can successively refine the value of $\widehat m$ and obtain a tighter result. Let us denote $\widehat m_i$ and $\widehat g_i (\cN_{D}^{\ox n}, \ve)$ as the values of $\widehat m$ and $\widehat g (\cN_{D}^{\ox n}, \ve)$ in the $i$-th iteration, respectively. First, we take the initial value of $\widehat m_{1} = g (\cN_{D}^{\ox n}, \ve)$ and get the result $\widehat g_1 (\cN_{D}^{\ox n}, \ve)$. Then we can set $\widehat m_{i+1} = \widehat g_i (\cN_{D}^{\ox n}, \ve)$ and get the result $\widehat g_{i+1} (\cN_{D}^{\ox n}, \ve)$. In Fig. \ref{g NS iso}, we show that after five iterations, we can get a converse bound $-\log \widehat g_5 (\cN_{D}^{\ox n}, \ve)$ which is strictly tighter than  $-\log f (\cN_{D}^{\ox n}, \ve)$. Especially, when $n=17$, $-\log \widehat g_5(\cN_{D}^{\ox n}, \ve) < 1 < -\log f(\cN_{D}^{\ox n}, \ve)$. It shows that we cannot transmit a single qubit within error tolerance $\ve = 0.004$ via $17$ copies of depolarizing channel where parameter $p=0.2$. However, this result cannot be obtained via the converse bound $-\log f(\cN_{D}^{\ox n}, \ve)$.

\begin{figure}[H]
\centering
\includegraphics[scale = 0.6]{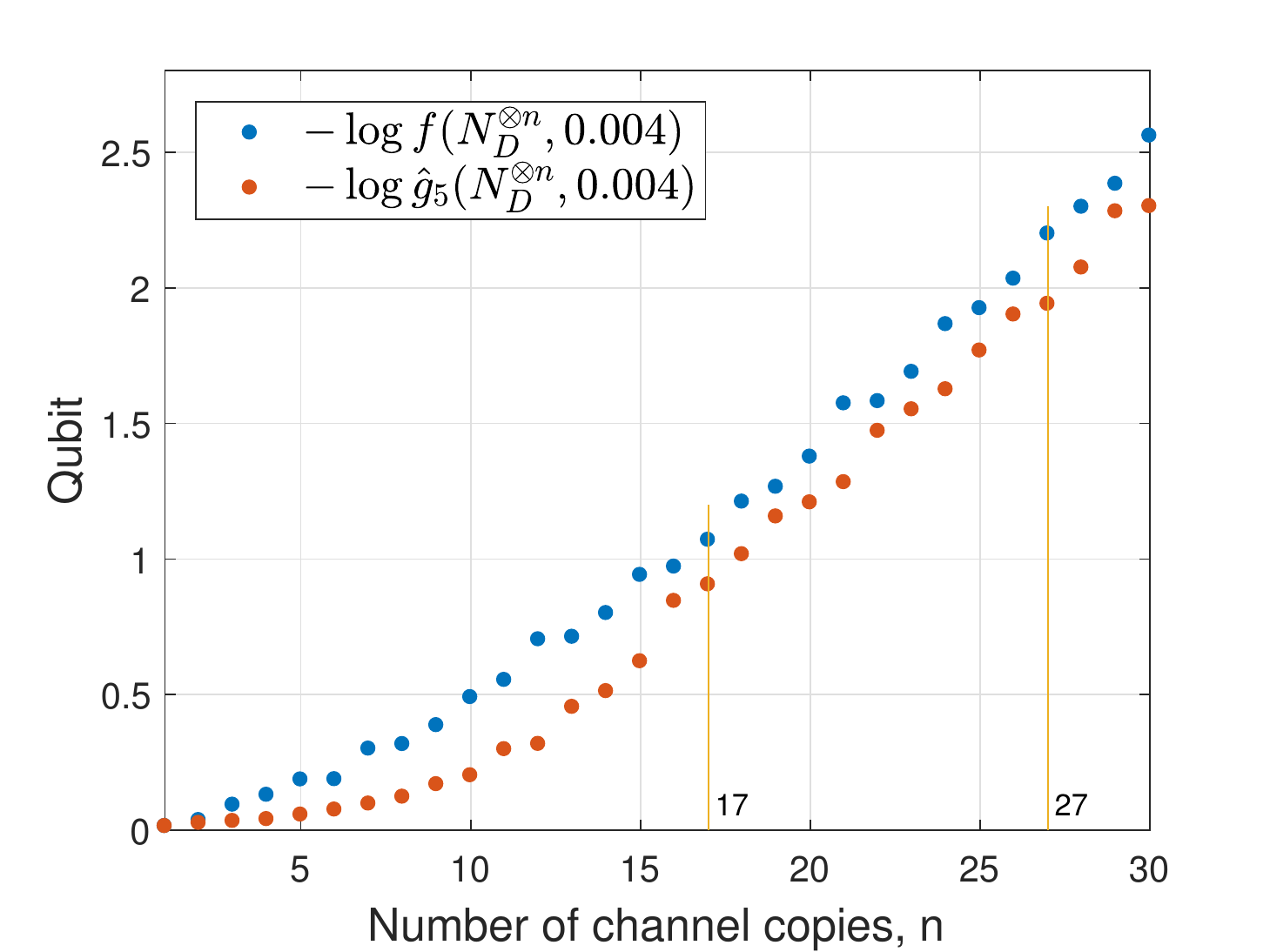}
\caption{This figure demonstrates the differences between the SDP converse bounds $-\log f (\cN_{D}^{\ox n}, 0.004)$ (blue dots) and $-\log \widehat g_5 (\cN_{D}^{\ox n}, 0.004)$ (red dots), where the channel parameter $p=0.2$ and the number of channel uses ranges from $1$ to $30$.}
\label{g NS iso}
\end{figure}


\section{Strong converse bound for quantum communication}
\label{Strong converse bound for quantum communication}

In this section, we establish an SDP strong converse bound $Q_\G$ (or $R_{\max}$) to evaluate the quantum capacity of a general quantum channel. We summarize our strong converse bound with other well-known bounds in Table \ref{compare table}. 

\subsection{An SDP strong converse bound on quantum capacity}

\begin{proposition}        
For any quantum channel $\cN$ and error tolerance $\ve$, 
\begin{align}
\label{oneshot q gamma bound}
Q_{\rm PPT}^{(1)}(\cN,\ve) \leq Q_\G(\cN) - \log (1-\ve),
\end{align}
where $Q_\G(\cN):=\log\G(\cN)$ and  
\begin{align}
\label{Gamma primal}
{\rm (Primal)} \quad & \G(\cN)= \max \left\{\tr  J_{\cN}R_{AB}: R_{AB},\rho_A\ge0, \tr{\rho_A}=1,
-\rho_A \ox \1_B \le R_{AB}^{T_{B}} \le \rho_A \ox \1_B\right\},\\{\rm (Dual)}\ \quad & \G(\cN)= \min \left\{\mu: Y_{AB},V_{AB}\ge0, (V_{AB}-Y_{AB})^{T_{B}} \ge J_{\cN},\tr_B(V_{AB}+Y_{AB})\le \mu \1_A\right\}. \label{dual WN}
\end{align}
\end{proposition}        
\begin{proof}
Suppose the optimal solution in the optimization (\ref{tradeoff sdp primal}) of $Q_{PPT}^{(1)}(\cN,\ve)$ is taken at $\{W_{AB},\rho_A,m\}$, then $Q_{PPT}^{(1)}(\cN,\ve) = -\log m$. Denote $R_{AB} = \frac1m W_{AB}$ and we can verify that $\{R_{AB},\rho_A\}$ is a feasible solution to the SDP~(\ref{Gamma primal}). Thus 
\[Q_\G(\cN) \geq \log \tr J_{\cN} R_{AB} = \log \frac1m\tr J_{\cN} W_{AB} \geq \log \frac1m (1-\ve) = Q_{PPT}^{(1)}(\cN,\ve) + \log (1-\ve).\]
This concludes the proof. 

The dual SDP can be derived via the Lagrange multiplier method. The main step is to associate a positive-semidefinite Lagrange multiplier for each inequality constraint. To be specific, we introduce $V_{AB},Y_{AB}\ge 0$ and a real multiplier $\mu$, and obtain the following Lagrangian:
\begin{equation}
\begin{split}
	&\tr J_\cN R_{AB}\\
	+&\tr (\rho_A\ox\1_B-R_{AB}^{T_B})V_{AB}\\
	+&\tr (\rho_A\ox\1_B+R_{AB}^{T_B})Y_{AB}\\
	+&\mu (1-\tr \rho_A)\\
	=&\mu+\tr R_{AB}(J_\cN-V_{AB}^{T_B}+Y_{AB}^{T_B})\\
	+&\tr \rho_{A}(\tr_B V_{AB}+\tr_B Y_{AB}-\mu \1_A).
\end{split}
\end{equation}
Hence, the dual SDP is to minimize $\mu$ subject to
\begin{align}
V_{AB},Y_{AB}&\ge 0,\\
J_\cN &\le V_{AB}^{T_B}-Y_{AB}^{T_B},\\
\tr_B (V_{AB}+Y_{AB})&\le\mu \1_A.
\end{align}
\end{proof}

\begin{proposition}
\label{Q gamma add}
For any quantum channel $\cN_1$ and $\cN_2$, we have
\begin{align}
Q_\G(\cN_1\ox \cN_2) = Q_\G(\cN_1) + Q_\G(\cN_2).
\end{align}
\end{proposition}
\begin{proof}    
We only need to show that $\G(\cN_1\ox \cN_2) = \G(\cN_1) \G(\cN_2)$. For the primal problem (\ref{Gamma primal}), suppose the optimal solutions of the SDP~(\ref{Gamma primal}) for $\cN_1$ and $\cN_2$ are 
$\{R_1,\rho_1\}$ and $\{R_2,\rho_2\}$, respectively. Then we can verify that 
$\{R_1\ox R_2, \rho_1 \ox \rho_2\}$ is a feasible solution of $\G(\cN_1\ox \cN_2)$. Thus $\G(\cN_1\ox \cN_2) \geq \tr (J_{\cN_1} \ox J_{\cN_2})(R_1\ox R_2) = \G(\cN_1) \G(\cN_2)$.

For the dual problem (\ref{dual WN}), suppose the optimal solutions of the SDP (\ref{dual WN}) for  $\cN_1$ and $\cN_2$ are $\{V_1,Y_1,\mu_1\}$ and $\{V_2,Y_2,\mu_2\}$. Denote $V = V_1\ox V_2 + Y_1 \ox Y_2$ and $Y = V_1 \ox Y_2 + Y_1 \ox V_2$. It can be easily verified that $\{V,Y,\mu_1 \mu_2\}$ is a feasible solution of $\G(\cN_1\ox \cN_2)$. Thus $\G(\cN_1\ox \cN_2) \leq \G(\cN_1)\G(\cN_2)$.
\end{proof}    

\begin{theorem}
For any quantum channel $\cN$, we have
\begin{align}
Q(\cN) \leq Q_{\rm PPT}(\cN) \leq Q_\G(\cN).
\end{align}
Moreover,  $Q_\G(\cN)$ is a strong converse bound. That is, if the rate exceeds $Q_\G(\cN)$, the error probability will approach to one exponentially fast as the number of channel uses increase.
\end{theorem}
\begin{proof}
We first show that $Q_\G(\cN)$ is a converse bound and then prove that it is a strong converse. From Eq. (\ref{oneshot q gamma bound}), take regularization on both sides, we have
\begin{equation}
\begin{split}
Q_{\rm PPT}(\cN) & = \lim_{\ve \rightarrow 0}\lim_{n\rightarrow \infty} \frac1n Q_{\rm PPT}^{(1)}(\cN^{\ox n},\ve)\\
& \leq \lim_{\ve \rightarrow 0}\lim_{n\rightarrow \infty} \frac1n\left[ Q_{\G}(\cN^{\ox n})-\log(1-\ve)\right]\\
& = Q_\G(\cN).
\end{split}
\end{equation}
In the last line, we use the additivity of $Q_\G$ in Proposition \ref{Q gamma add}.

For the quantum channel $\cN^{\ox n}$, suppose its achievable rate is $r$. From Eq. (\ref{oneshot q gamma bound}), we know that $nr \leq n Q_\G(\cN) - \log (1-\ve)$, which implies
\begin{align}
\ve \geq 1-2^{n(Q_\G(\cN) - r)}.
\end{align}
If $r > Q_\G(\cN)$, the error will exponentially converge to one as $n$ increases.
\end{proof}

\vspace{0.2cm}
\begin{remark}
For $d$-dimensional noiseless quantum channel $\cI_d$, we can show $Q(\cI_d) = Q_\G(\cI_d) = \log d$.
\end{remark}

\subsection{Comparison with other converse bounds}

There are several well-known converse bounds on quantum capacity. In this subsection, we compare them with our SDP strong converse bound $Q_\G$. 
\begin{table}[H] 
\centering
\includegraphics[scale=1]{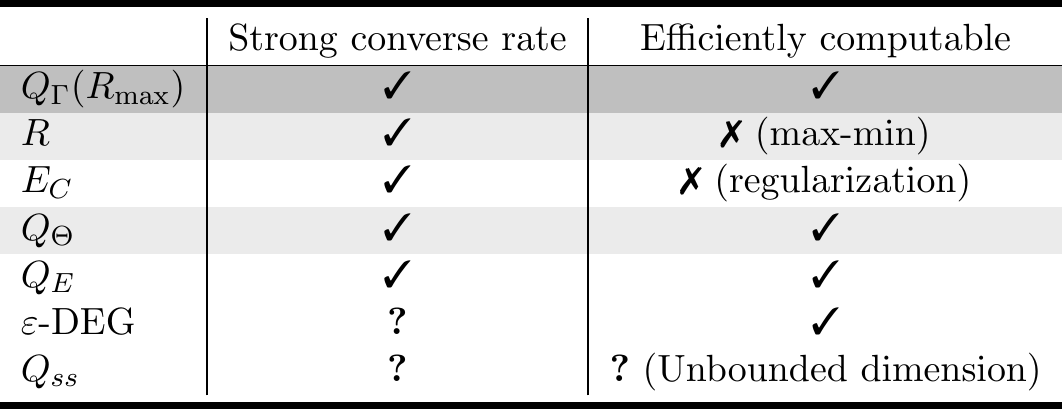}
\caption{Comparison of converse bounds on quantum capacity. 
}
\label{compare table}
\end{table} 

\textit{Tomamichel et al.}  \cite{Tomamichel2015a} established that the Rains information of any quantum channel is a strong converse rate for quantum communication. To be specific, the Rains information of a quantum channel is defined as \cite{Tomamichel2015a}:
\begin{alignat}{1}
R(\cN) :=\max_{\rho_A \in \cS(A)} \min_{\sigma_{AB} \in \text{PPT'}} D(\cN_{A'\rightarrow B}(\phi_{AA'}) \big\| \sigma_{AB}) ,
\end{alignat}
where $\phi_{AA'}$ is a purification of $\rho_A$ and the set $\text{PPT'} = \{\sigma \in \cP(A\ox B): \big\|\sigma^{T_B}\big\|_1 \leq 1\}$.
We note that our bound $Q_\Gamma$ is weaker than the Rains information (cf.~Corollary~\ref{QGamma Rains}).
However, $R(\cN)$ is not known to be efficiently computable for general quantum channels since it is a max-min optimization problem.

An efficiently computable converse bound (abbreviated as $\ve$-DEG) is given by the concept of approximate degradable channel~\cite{Sutter2014}. This bound usually works very well for approximate degradable quantum channels such as low-noise qubit depolarizing channel. See Ref.~\cite{Leditzky2017a,Sharma2017}  for some recent works based on this approach. Otherwise, it will degenerate to a trivial upper bound. We can easily show an example that $Q_\G$ can be smaller than $\ve$-DEG bound, e.g., the channel $\cN_r$ in Eq. (\ref{channel NR}) with $0<r<0.38$. Also, it is unknown whether $\ve$-DEG bound is a strong converse.

Another previously known efficiently computable strong converse bound for general channels is given by the partial transposition bound \cite{Holevo2001,Muller-Hermes2015},
\begin{align}
Q_{\Theta}(\cN):=\log \left\| \cN\circ T\right\|_{\di},
\end{align} 
where $T$ is the transpose map and  $\|\cdot\|_{\di}$ is the completely bounded trace norm. Note that which $\|\cdot\|_{\di}$ is known
to be efficiently computable via semidefinite programming in Ref.~\cite{Watrous2012}. 

The entanglement cost of a quantum channel~\cite{Berta2013}, denoted as $E_C$, is proved to be a strong converse bound. But it is not known to be efficiently computable for general channels, due to its regularization. The entanglement-assisted quantum  $Q_E$ is also a strong converse for the quantum capacity~\cite{Bennett2014,Berta2011a} and there is a recently developed approach to efficiently compute it~\cite{Fawzi2017a}. Quantum capacity with symmetric side channels~\cite{Smith2008a}, denoted as $Q_{ss}$, is also an important converse bound for general channels. But it is not known to be computable due to the potentially unbounded dimension of the side channel. It is also not known to be a strong converse.



\begin{theorem}
For any quantum channel $\cN$, we have
\begin{equation}
\label{inequality chain}
Q(\cN) \le {R}(\mathcal{N})\le Q_{\G}(\cN) \le Q_{\Theta}(\cN).\end{equation}
\end{theorem}
The first inequality has been proved in Ref.~\cite{Tomamichel2015a}. 
We prove the second inequality in Corollary \ref{QGamma Rains} and the third inequality in Proposition \ref{QGmmma QTheta}.

\vspace{0.5cm}
In the following proof, we need to introduce an entanglement measure $E_W$ which is defined in Ref.~\cite{Wang2016}. We will see that the strong converse bound $Q_\G$ is a channel analogue of entanglement measure $E_W$ and can be further reformulated into a similar form as the Rains information.  Specifically, for any bipartite quantum state $\rho_{AB}$, the entanglement measure $E_W$ is defined by $E_W(\rho) :=\log W(\rho)$, where
\begin{alignat}{2}
& {(\rm Primal)} \quad && W(\rho)= \max \left\{ \tr  \rho R_{AB}: \left|R_{AB}^{T_{B}}\right| \le  \1, R_{AB}\ge0\right\}, \label{EW primal}\\
& \ {(\rm Dual)} && W(\rho)= \min \left\{\left\|X_{AB}^{T_B} \right\|_1:  X_{AB} \geq \rho_{AB}\right\}. \label{EW dual}
\end{alignat}
The max-relative entropy of two operators $\rho \in \cS_\leq (A)$, $\sigma \in \cP(A)$ is defined by~\cite{Datta2009}
\begin{align}
D_{\max}(\rho\|\sigma):= \log \min \{\mu: \rho \leq \mu \sigma\}.
\end{align}

\begin{proposition}
\label{QGamma EW}
For any quantum channel $\cN$, it holds that
\begin{align}Q_\G(\cN) = \max_{\rho_A \in \cS(A)} E_W(\cN_{A'\rightarrow B}(\phi_{AA'}))=\max_{\rho \in S(A)} \min_{\sigma \in \rm{PPT'}} D_{\max}(\cN_{A'\to B}( \phi_{AA'})\big\|\sigma_{AB}),\end{align}
where  $\phi_{AA'}$ is a purification of $\rho_A$ and  the set $\text{PPT'} = \left\{\sigma \in \cP(A\ox B): \big\|\sigma^{T_B}\big\|_1 \leq 1 \right\}$.
\end{proposition}
\begin{proof}
Consider purification $\phi_{AA'} = \rho_{A}^{1/2}\Phi_{AA'}\rho_{A}^{1/2} (= \rho_{A'}^{1/2}\Phi_{AA'}\rho_{A'}^{1/2})$, then 
\begin{align}
\cN_{A'\rightarrow B}(\phi_{AA'}) = \cN_{A'\rightarrow B}(\rho_{A}^{1/2}\Phi_{AA'}\rho_{A}^{1/2}) = \rho_{A}^{1/2}\cN_{A'\rightarrow B}(\Phi_{AA'})\rho_{A}^{1/2} = \rho_{A}^{1/2}J_{\cN}\rho_{A}^{1/2}.\end{align}

Take $J_\cN = \rho_{A}^{-1/2}\cN_{A'\rightarrow B}(\phi_{AA'}) \rho_{A}^{-1/2}$ into the definition of $Q_\G(\cN)$ (\ref{Gamma primal}) and substitute $F_{AB} = \rho_{A}^{-1/2} R_{AB}\rho_{A}^{-1/2}$, we have
\begin{equation}\begin{split}
Q_\G(\cN)=\log  \max  & \tr \cN_{A'\rightarrow B}(\phi_{AA'}) F_{AB}\\
\text{s.t.} & \ F_{AB},\rho_{A}\ge0, \tr{\rho_{A}}=1, -\1_{AB} \le F_{AB}^{T_{B}} \le \1_{AB}
\end{split}\end{equation}
Note that here we only consider invertible state $\rho_{A}$. The reason is that the set of invertible positive operators is dense in the set of all positive semidefinite operators, and then it suffices to optimize with respect to them. 

Due to the definition of $E_W$ (\ref{EW primal}), we have 
\begin{equation}\begin{split}
Q_\G(\cN)=\max_{\rho_{A} \in \cS(A)} E_W(\cN_{A'\rightarrow B}(\phi_{AA'})).
\end{split}\end{equation}

On the other hand, the following equality chain holds
\begin{equation}
\begin{split}
E_W(\rho)&=\log \min \left\{\big\|X^{T_B}\big\|_1: \rho\le X\right\}\\
&=\log \min \left\{\mu :\rho \le X, \big\|X^{T_B}\big\|_1\leq \mu\right\}\\
&=\log \min \left\{\mu :\rho \le \mu \sigma,  \big\|\mu \sigma^{T_B}\big\|_1\leq \mu\right\}\\
&=\log \min \left\{\mu: \rho \le \mu \sigma,  \big\|\sigma^{T_B}\big\|_1\leq 1 \right\}\\
&= \min_{\sigma\in \text{PPT'}} D_{\max}(\rho\|\sigma).
\end{split}
\end{equation}
The first line follows from Eq. (\ref{EW dual}). In the second line, we introduce a new variable $\mu$. In the third line, we substitute $X$ with $\mu \sigma$. The last line follows from the definition of $D_{\max}$. This directly implies that $E_W(\rho)\ge R(\rho)$ (note, also \footnote{We note that Andreas Winter told us the fact that $E_W$ can be proved to be an upper bound of the Rains bound by some optimization techniques in the past.}).

Therefore, we have that
\begin{align}
Q_\G(\cN) &= \max_{\rho_A \in \cS(A)} E_W(\cN_{A'\rightarrow B}(\phi_{AA'}))\\
&=\max_{\rho \in S(A)} \min_{\sigma \in \rm{PPT'}} D_{\max}(\cN_{A'\to B}( \phi_{A'A})\big\|\sigma_{AB}).
\end{align}
\end{proof}

We note that the max-relative entropy of entanglement of a quantum channel~\cite{Christandl2016} and our bound $Q_\G$ are in the same spirit, but for evaluating quantum communication, the bound in~\cite{Christandl2016} is weaker than  the bound $Q_\Gamma$ as well as the Rains information \cite{Tomamichel2015a}.

\begin{corollary}
\label{QGamma Rains}
For any quantum channel $\cN$, it holds that 
\begin{align}
{R}(\mathcal{N})\le Q_{\G}(\cN).
\end{align}    
\end{corollary}
\begin{proof}
Note that $D(\rho\|\sigma) \leq D_{\max}(\rho\|\sigma)$~\cite{Datta2009}, we have 
\begin{equation}
\begin{split}
Q_\G(\cN) & = \max_{\rho \in S(A)} \min_{\sigma \in \rm{PPT'}} D_{\max}(\cN_{A'\to B}( \phi_{A'A})\big\|\sigma_{AB})\\
&\geq \max_{\rho_A \in \cS(A)} \min_{\sigma \in \text{PPT'}} D(\cN_{A'\rightarrow B}(\phi_{AA'}) \big\| \sigma_{AB}) = R(\cN).
\end{split}
\end{equation} 
\end{proof}

\vspace{-0.5cm}
\begin{proposition}
\label{QGmmma QTheta}
For any quantum channel $\cN$, it holds that $Q_{\G}(\cN) \le Q_{\Theta}(\cN)$.
\end{proposition}
\begin{proof}
Suppose the optimal solution of SDP (\ref{Gamma primal}) is taken at $\{R_{AB},\rho_A\}$, then 
$\G(\cN)=\tr J_{\cN} R_{AB}=\tr J_{\cN}^{T_B} R_{AB}^{T_B}$. The completely bounded trace norm can be written as SDP~\cite{Watrous2012}, 
\begin{align}\label{cb norm}
\left\| \cN \circ T \right\|_{\di}= \max \left\{ \frac{1}{2}\tr J_{\cN}^{T_B} (X+X^\dagger):
\left( {\begin{array}{*{20}{c}}
	{\rho_0 \otimes \1 }&X\\
	X^\dagger&{\rho_1 \otimes \1  }
	\end{array}} \right)\ge 0,\ \rho_0, \rho_1 \in \cS(A).\right\}
\end{align}
Since $-\rho_A \ox \1_B\le R_{AB}^{T_B}\le \rho_A \ox \1_B$, we have
\begin{align}
\begin{pmatrix}
\rho_A \otimes \1_B &R_{AB}^{T_B}\\
R_{AB}^{T_B}& \rho_A \otimes \1_B
\end{pmatrix}
=\frac{1}{2}\begin{pmatrix}
1&1\\1&1 \end{pmatrix} 
\bigotimes (\rho_A\otimes \1+R_{AB}^{T_B})
+\frac{1}{2}\begin{pmatrix}
1&-1\\-1&1\end{pmatrix} 
\bigotimes (\rho_A \otimes \1-R_{AB}^{T_B})\ge 0.
\end{align}
So $\left\{R_{AB}^{T_B},\rho_A,\rho_A\right\}$ is a feasible solution of SDP (\ref{cb norm}), which means that
\begin{align}
Q_\Theta(\cN) = \log \left\| \cN\circ T \right\|_{\di}\ge \log \tr (J_{\cN}^{T_B} R_{AB}^{T_B})= \log \G(\cN) = Q_\G(\cN).
\end{align}
\end{proof}

\vspace{0.3cm}
In Fig.~\ref{spec3}, we compare the converse bound $Q_{\G}$ with $Q_{\Theta}$  in the case of quantum channel 
\begin{equation}\label{channel NR}
\cN_r=\sum_{i=0}^1 E_i\cdot E_i^\dagger, 
\end{equation}
where $E_0=\proj 0+\sqrt r\proj 1$ and $E_1=\sqrt{1-r}\ketbra 0 1+\ketbra 1 2$ $(0\le r\le 0.5)$. In the following Fig. \ref{spec3}, it is clear that $Q_{\G}(\cN)$ can be strictly tighter than $Q_{\Theta}(\cN)$.

\begin{figure}[H]
\center
\includegraphics[scale = 0.6]{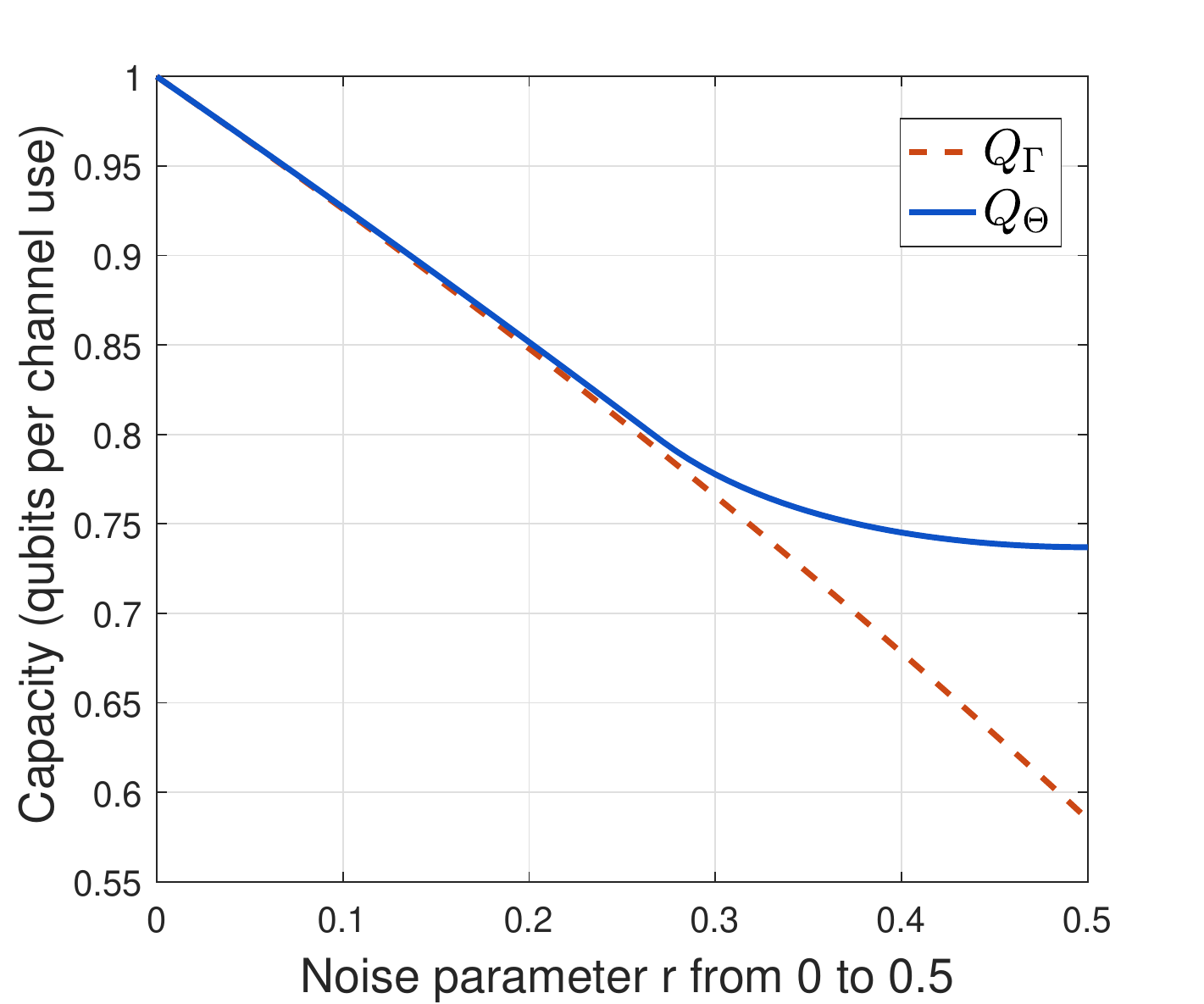}
\caption{This plot demonstrates the difference between converse bounds $Q_\G(\cN_r)$ and $Q_\Theta(\cN_r)$. The dashed line depicts $Q_\G(\cN_r)$ while the solid line depicts $Q_\Theta(\cN_r)$. The parameter $r$ ranges from $0$ to $0.5$.}
\label{spec3} 
\end{figure}

\section{Discussions}
In summary, we have derived efficiently computable converse bounds to evaluate the capabilities of quantum communication over quantum channels in both the non-asymptotic and asymptotic settings by utilizing the techniques of convex optimization.

We have provided one-shot converse bounds in the context of quantum communication with finite resources, which improves the previous general SDP converse bound in~\cite{Tomamichel2016}. Furthermore, in the asymptotic regime, we have derived an SDP strong converse bound $Q_\Gamma$ for quantum communication, which is better than the partial transpose bound~\cite{Holevo2001} as well as the max-Relative entropy of a channel \cite{Christandl2016}. Furthermore, we have refined the $Q_\Gamma$ as the so-called max-Rains information via connecting it to the SDP entanglement measure in~\cite{Wang2016}. It is worth noting that our bound is no better than the Rains information~\cite{Tomamichel2015a} in general, but it is the best SDP-computable strong converse bound. It is also worth noting that our bound $Q_\G$ was recently proved to be a strong converse bound for the LOCC-assisted quantum capacity in~\cite{Berta2017a}.

However, for the qubit depolarizing channel, the bound $Q_\G$ does not work very well. The best to date converse bound of this particular channel is still given by~\cite{Sutter2014,Leditzky2017,Smith2008b}. 
It is of great interest to use the one-shot SDP converse bound in Eq.~(\ref{PPT NS hat g}) to provide a potentially better upper bound on the quantum capacity of depolarizing channel. 

\section*{Acknowledgments}
We were grateful to Mario Berta, Felix Leditzky, Debbie Leung, Marco Tomamichel, Andreas Winter and Mark M. Wilde for helpful discussions.  XW
acknowledges support from the Department of Defense.
This work was mainly done when XW and RD were at UTS  and it was partly supported by the Australian Research Council (Grant No. DP120103776 and No. FT120100449). 

\bibliographystyle{IEEEtran}
\bibliography{Bib}

\begin{thebibliography}{10}
\providecommand{\url}[1]{#1}
\csname url@samestyle\endcsname
\providecommand{\newblock}{\relax}
\providecommand{\bibinfo}[2]{#2}
\providecommand{\BIBentrySTDinterwordspacing}{\spaceskip=0pt\relax}
\providecommand{\BIBentryALTinterwordstretchfactor}{4}
\providecommand{\BIBentryALTinterwordspacing}{\spaceskip=\fontdimen2\font plus
\BIBentryALTinterwordstretchfactor\fontdimen3\font minus
  \fontdimen4\font\relax}
\providecommand{\BIBforeignlanguage}[2]{{%
\expandafter\ifx\csname l@#1\endcsname\relax
\typeout{** WARNING: IEEEtran.bst: No hyphenation pattern has been}%
\typeout{** loaded for the language `#1'. Using the pattern for}%
\typeout{** the default language instead.}%
\else
\language=\csname l@#1\endcsname
\fi
#2}}
\providecommand{\BIBdecl}{\relax}
\BIBdecl

\bibitem{Wang2016a}
X.~Wang and R.~Duan, ``{A semidefinite programming upper bound of quantum
  capacity},'' in \emph{2016 IEEE International Symposium on Information Theory
  (ISIT)}, vol. 2016-Augus.\hskip 1em plus 0.5em minus 0.4em\relax IEEE, July
  2016, pp. 1690--1694.

\bibitem{Lloyd1997}
\BIBentryALTinterwordspacing
S.~Lloyd, ``{Capacity of the noisy quantum channel},'' \emph{Physical Review
  A}, vol.~55, no.~3, pp. 1613--1622, March 1997. 
\BIBentrySTDinterwordspacing

\bibitem{Shor2002a}
P.~W. Shor, ``{The quantum channel capacity and coherent information},'' in
  \emph{lecture notes, MSRI Workshop on Quantum Computation}, 2002.

\bibitem{Devetak2005a}
\BIBentryALTinterwordspacing
I.~Devetak, ``{The Private Classical Capacity and Quantum Capacity of a Quantum
  Channel},'' \emph{IEEE Transactions on Information Theory}, vol.~51, no.~1,
  pp. 44--55, January 2005. 
  \BIBentrySTDinterwordspacing

\bibitem{Schumacher1996a}
\BIBentryALTinterwordspacing
B.~Schumacher and M.~A. Nielsen, ``{Quantum data processing and error
  correction},'' \emph{Physical Review A}, vol.~54, no.~4, pp. 2629--2635, October
  1996. 
\BIBentrySTDinterwordspacing

\bibitem{Barnum2000}
\BIBentryALTinterwordspacing
H.~Barnum, E.~Knill, and M.~Nielsen, ``{On quantum fidelities and channel
  capacities},'' \emph{IEEE Transactions on Information Theory}, vol.~46,
  no.~4, pp. 1317--1329, July 2000. 
\BIBentrySTDinterwordspacing

\bibitem{Barnum1998}
\BIBentryALTinterwordspacing
H.~Barnum, M.~A. Nielsen, and B.~Schumacher, ``{Information transmission
  through a noisy quantum channel},'' \emph{Physical Review A}, vol.~57, no.~6,
  pp. 4153--4175, June 1998. 
\BIBentrySTDinterwordspacing

\bibitem{Cubitt2015}
\BIBentryALTinterwordspacing
T.~Cubitt, D.~Elkouss, W.~Matthews, M.~Ozols, D.~P{\'{e}}rez-Garc{\'{i}}a, and
  S.~Strelchuk, ``{Unbounded number of channel uses may be required to detect
  quantum capacity},'' \emph{Nature Communications}, vol.~6, no.~1, p. 6739,
  December 2015. 
\BIBentrySTDinterwordspacing

\bibitem{Elkouss2015}
\BIBentryALTinterwordspacing
D.~Elkouss and S.~Strelchuk, ``{Superadditivity of Private Information for Any
  Number of Uses of the Channel},'' \emph{Physical Review Letters}, vol. 115,
  no.~4, p. 040501, July 2015. \BIBentrySTDinterwordspacing

\bibitem{DiVincenzo1998a}
\BIBentryALTinterwordspacing
D.~P. DiVincenzo, P.~W. Shor, and J.~A. Smolin, ``{Quantum-channel capacity of
  very noisy channels},'' \emph{Physical Review A}, vol.~57, no.~2, pp.
  830--839, February 1998. 
\BIBentrySTDinterwordspacing

\bibitem{Fern2008}
\BIBentryALTinterwordspacing
J.~Fern and K.~B. Whaley, ``{Lower bounds on the nonzero capacity of Pauli
  channels},'' \emph{Physical Review A}, vol.~78, no.~6, p. 062335, December 2008.
\BIBentrySTDinterwordspacing

\bibitem{Smith2007}
\BIBentryALTinterwordspacing
G.~Smith and J.~A. Smolin, ``{Degenerate Quantum Codes for Pauli Channels},''
  \emph{Physical Review Letters}, vol.~98, no.~3, p. 030501, January 2007.
\BIBentrySTDinterwordspacing

\bibitem{Smith2008a}
\BIBentryALTinterwordspacing
G.~Smith, J.~A. Smolin, and A.~Winter, ``{The Quantum Capacity With Symmetric
  Side Channels},'' \emph{IEEE Transactions on Information Theory}, vol.~54,
  no.~9, pp. 4208--4217, September 2008. 
\BIBentrySTDinterwordspacing

\bibitem{Sutter2014}
\BIBentryALTinterwordspacing
D.~Sutter, V.~B. Scholz, A.~Winter, and R.~Renner, ``{Approximate Degradable
  Quantum Channels},'' \emph{IEEE Transactions on Information Theory}, vol.~63,
  no.~12, pp. 7832--7844, December 2017. 
\BIBentrySTDinterwordspacing

\bibitem{Leditzky2017}
\BIBentryALTinterwordspacing
F.~Leditzky, N.~Datta, and G.~Smith, ``{Useful states and entanglement
  distillation},'' \emph{IEEE Transactions on Information Theory},  vol.~64,
  no.~7, pp. 4689--4708, July 2018. 
\BIBentrySTDinterwordspacing

\bibitem{Leung2015a}
\BIBentryALTinterwordspacing
D.~Leung and J.~Watrous, ``{On the complementary quantum capacity of the
  depolarizing channel},'' \emph{Quantum}, vol.~1, p.~28, September 2017. 
\BIBentrySTDinterwordspacing

\bibitem{Wolfowitz1978}
J.~Wolfowitz, ``{Coding theorems of information theory},'' \emph{Mathematics of
  Computation}, 1978.

\bibitem{Ogawa1999}
\BIBentryALTinterwordspacing
T.~Ogawa and H.~Nagaoka, ``{Strong converse to the quantum channel coding
  theorem},'' \emph{IEEE Transactions on Information Theory}, vol.~45, no.~7,
  pp. 2486--2489, 1999. 
\BIBentrySTDinterwordspacing

\bibitem{Winter1999}
A.~Winter, ``{Coding theorem and strong converse for quantum channels},''
  \emph{IEEE Transactions on Information Theory}, vol.~45, no.~7, pp.
  2481--2485, 1999.

\bibitem{Koenig2009}
\BIBentryALTinterwordspacing
R.~K{\"{o}}nig and S.~Wehner, ``{A Strong Converse for Classical Channel Coding
  Using Entangled Inputs},'' \emph{Physical Review Letters}, vol. 103, no.~7,
  p. 070504, August 2009. 
  \BIBentrySTDinterwordspacing

\bibitem{Wilde2013a}
M.~M. Wilde and A.~Winter, ``{Strong converse for the classical capacity of the
  pure-loss bosonic channel},'' \emph{Problems of Information Transmission},
  vol.~50, no.~2, pp. 117--132, 2013.

\bibitem{Wilde2014a}
\BIBentryALTinterwordspacing
M.~M. Wilde, A.~Winter, and D.~Yang, ``{Strong Converse for the Classical
  Capacity of Entanglement-Breaking and Hadamard Channels via a Sandwiched
  R{\'{e}}nyi Relative Entropy},'' \emph{Communications in Mathematical
  Physics}, vol. 331, no.~2, pp. 593--622, October 2014. 
\BIBentrySTDinterwordspacing

\bibitem{Wang2016g}
\BIBentryALTinterwordspacing
X.~Wang, W.~Xie, and R.~Duan, ``{Semidefinite programming strong converse
  bounds for classical capacity},'' \emph{IEEE Transactions on Information
  Theory}, vol.~64, no.~1, pp. 640--653, October 2018. 
\BIBentrySTDinterwordspacing

\bibitem{Tomamichel2015a}
M.~Tomamichel, M.~M. Wilde, and A.~Winter, ``{Strong Converse Rates for Quantum
  Communication},'' \emph{IEEE Transactions on Information Theory}, vol.~63,
  no.~1, pp. 715--727, January 2017.

\bibitem{Holevo2001}
\BIBentryALTinterwordspacing
A.~Holevo and R.~Werner, ``{Evaluating capacities of bosonic Gaussian
  channels},'' \emph{Physical Review A}, vol.~63, no.~3, p. 032312, February 2001.
  \BIBentrySTDinterwordspacing

\bibitem{Muller-Hermes2015}
\BIBentryALTinterwordspacing
A.~M{\"{u}}ller-Hermes, D.~Reeb, and M.~M. Wolf, ``{Positivity of linear maps
  under tensor powers},'' \emph{Journal of Mathematical Physics}, vol.~57,
  no.~1, p. 015202, January 2016. 
\BIBentrySTDinterwordspacing

\bibitem{Gao2015a}
L.~Gao, M.~Junge, and N.~LaRacuente, ``{Capacity Bounds via Operator Space
  Methods},'' \emph{arXiv:1509.07294}, 2015.

\bibitem{Bruß1998}
\BIBentryALTinterwordspacing
D.~Bru{\ss}, D.~P. DiVincenzo, A.~Ekert, C.~A. Fuchs, C.~Macchiavello, and
  J.~A. Smolin, ``{Optimal universal and state-dependent quantum cloning},''
  \emph{Physical Review A}, vol.~57, no.~4, pp. 2368--2378, April 1998. 
\BIBentrySTDinterwordspacing

\bibitem{Cerf2000}
\BIBentryALTinterwordspacing
N.~J. Cerf, ``{Pauli Cloning of a Quantum Bit},'' \emph{Physical Review
  Letters}, vol.~84, no.~19, pp. 4497--4500, May 2000. \BIBentrySTDinterwordspacing

\bibitem{Wolf2007}
\BIBentryALTinterwordspacing
M.~M. Wolf and D.~P{\'{e}}rez-Garc{\'{i}}a, ``{Quantum capacities of channels
  with small environment},'' \emph{Physical Review A}, vol.~75, no.~1, p.
  012303, January 2007. 
  \BIBentrySTDinterwordspacing

\bibitem{Smith2008b}
G.~Smith and J.~A. Smolin, ``{Additive extensions of a quantum channel},'' in
  \emph{Proceedings of IEEE Information Theory Workshop (ITW)}.\hskip 1em plus
  0.5em minus 0.4em\relax IEEE, 2008, pp. 368--372.

\bibitem{Hayashi2009}
\BIBentryALTinterwordspacing
M.~Hayashi, ``{Information Spectrum Approach to Second-Order Coding Rate in
  Channel Coding},'' \emph{IEEE Transactions on Information Theory}, vol.~55,
  no.~11, pp. 4947--4966, November 2009. 
\BIBentrySTDinterwordspacing

\bibitem{Polyanskiy2010}
\BIBentryALTinterwordspacing
Y.~Polyanskiy, H.~V. Poor, and S.~Verdu, ``{Channel Coding Rate in the Finite
  Blocklength Regime},'' \emph{IEEE Transactions on Information Theory},
  vol.~56, no.~5, pp. 2307--2359, May 2010. 
\BIBentrySTDinterwordspacing

\bibitem{Tomamichel2013a}
\BIBentryALTinterwordspacing
M.~Tomamichel and M.~Hayashi, ``{A Hierarchy of Information Quantities for
  Finite Block Length Analysis of Quantum Tasks},'' \emph{IEEE Transactions on
  Information Theory}, vol.~59, no.~11, pp. 7693--7710, November 2013.
  \BIBentrySTDinterwordspacing

\bibitem{Wang2012}
\BIBentryALTinterwordspacing
L.~Wang and R.~Renner, ``{One-Shot Classical-Quantum Capacity and Hypothesis
  Testing},'' \emph{Physical Review Letters}, vol. 108, no.~20, p. 200501, May
  2012. 
\BIBentrySTDinterwordspacing

\bibitem{Leung2015c}
D.~Leung and W.~Matthews, ``{On the Power of PPT-Preserving and Non-Signalling
  Codes},'' \emph{IEEE Transactions on Information Theory}, vol.~61, no.~8, pp.
  4486--4499, August 2015.

\bibitem{Matthews2014}
\BIBentryALTinterwordspacing
W.~Matthews and S.~Wehner, ``{Finite Blocklength Converse Bounds for Quantum
  Channels},'' \emph{IEEE Transactions on Information Theory}, vol.~60, no.~11,
  pp. 7317--7329, November 2014. 
\BIBentrySTDinterwordspacing

\bibitem{Tomamichel2015b}
\BIBentryALTinterwordspacing
M.~Tomamichel, \emph{{Quantum Information Processing with Finite Resources}},
  ser. SpringerBriefs in Mathematical Physics.\hskip 1em plus 0.5em minus
  0.4em\relax Cham: Springer International Publishing, 2016, vol.~5. 
  \BIBentrySTDinterwordspacing

\bibitem{Beigi2015}
\BIBentryALTinterwordspacing
S.~Beigi, N.~Datta, and F.~Leditzky, ``{Decoding quantum information via the
  Petz recovery map},'' \emph{Journal of Mathematical Physics}, vol.~57, no.~8,
  p. 082203, August 2016. 
\BIBentrySTDinterwordspacing

\bibitem{Tomamichel2016}
\BIBentryALTinterwordspacing
M.~Tomamichel, M.~Berta, and J.~M. Renes, ``{Quantum coding with finite
  resources},'' \emph{Nature Communications}, vol.~7, p. 11419, May 2016.
\BIBentrySTDinterwordspacing

\bibitem{Wang2016}
\BIBentryALTinterwordspacing
X.~Wang and R.~Duan, ``{Improved semidefinite programming upper bound on
  distillable entanglement},'' \emph{Physical Review A}, vol.~94, no.~5, p.
  050301, November 2016.
\BIBentrySTDinterwordspacing

\bibitem{Rains2001}
\BIBentryALTinterwordspacing
E.~Rains, ``{A semidefinite program for distillable entanglement},'' \emph{IEEE
  Transactions on Information Theory}, vol.~47, no.~7, pp. 2921--2933, 2001.
\BIBentrySTDinterwordspacing

\bibitem{Vandenberghe1996}
L.~Vandenberghe and S.~Boyd, ``{Semidefinite Programming},'' \emph{SIAM
  Review}, vol.~38, no.~1, pp. 49--95, March 1996.

\bibitem{Grant2008}
\BIBentryALTinterwordspacing
M.~Grant and S.~Boyd, ``{CVX: Matlab software for disciplined convex
  programming},'' 2008. 
  \BIBentrySTDinterwordspacing

\bibitem{NathanielJohnston2016}
\BIBentryALTinterwordspacing
{Nathaniel Johnston}, ``{QETLAB: A MATLAB toolbox for quantum entanglement,
  version 0.9},'' 2016. 
\BIBentrySTDinterwordspacing

\bibitem{Chiribella2008}
\BIBentryALTinterwordspacing
G.~Chiribella, G.~M. D'Ariano, and P.~Perinotti, ``{Transforming quantum
  operations: Quantum supermaps},'' \emph{EPL (Europhysics Letters)}, vol.~83,
  no.~3, p. 30004, August 2008.
\BIBentrySTDinterwordspacing

\bibitem{Duan2016}
\BIBentryALTinterwordspacing
R.~Duan and A.~Winter, ``{No-Signalling-Assisted Zero-Error Capacity of Quantum
  Channels and an Information Theoretic Interpretation of the Lov{\'{a}}sz
  Number},'' \emph{IEEE Transactions on Information Theory}, vol.~62, no.~2,
  pp. 891--914, February 2016. 
\BIBentrySTDinterwordspacing

\bibitem{Eggeling2002a}
\BIBentryALTinterwordspacing
T.~Eggeling, D.~Schlingemann, and R.~F. Werner, ``{Semicausal operations are
  semilocalizable},'' \emph{Europhysics Letters (EPL)}, vol.~57, no.~6, pp.
  782--788, March 2002. 
\BIBentrySTDinterwordspacing

\bibitem{Cubitt2011}
\BIBentryALTinterwordspacing
T.~S. Cubitt, D.~Leung, W.~Matthews, and A.~Winter, ``{Zero-Error Channel
  Capacity and Simulation Assisted by Non-Local Correlations},'' \emph{IEEE
  Transactions on Information Theory}, vol.~57, no.~8, pp. 5509--5523, August
  2011. 
\BIBentrySTDinterwordspacing

\bibitem{Leditzky2017a}
\BIBentryALTinterwordspacing
F.~Leditzky, D.~Leung, and G.~Smith, ``{Quantum and Private Capacities of
  Low-Noise Channels},'' \emph{Physical Review Letters}, vol. 120, no.~16, p.
  160503, April 2018. 
  \BIBentrySTDinterwordspacing

\bibitem{Sharma2017}
\BIBentryALTinterwordspacing
K.~Sharma, M.~M. Wilde, S.~Adhikari, and M.~Takeoka, ``{Bounding the
  energy-constrained quantum and private capacities of phase-insensitive
  bosonic Gaussian channels},'' \emph{New Journal of Physics}, vol.~20, no.~6,
  p. 063025, June 2018.
\BIBentrySTDinterwordspacing

\bibitem{Watrous2012}
J.~Watrous, ``{Simpler semidefinite programs for completely bounded norms},''
  \emph{Chicago Journal of Theoretical Computer Science}, vol.~19, no.~1, pp.
  1--19, 2013.

\bibitem{Berta2013}
\BIBentryALTinterwordspacing
M.~Berta, F.~G. S.~L. Brandao, M.~Christandl, and S.~Wehner, ``{Entanglement
  Cost of Quantum Channels},'' \emph{IEEE Transactions on Information Theory},
  vol.~59, no.~10, pp. 6779--6795, October 2013.
\BIBentrySTDinterwordspacing

\bibitem{Bennett2014}
C.~H. Bennett, I.~Devetak, A.~W. Harrow, P.~W. Shor, and A.~Winter, ``{The
  Quantum Reverse Shannon Theorem and Resource Tradeoffs for Simulating Quantum
  Channels},'' \emph{IEEE Transactions on Information Theory}, vol.~60, no.~5,
  pp. 2926--2959, 2014.

\bibitem{Berta2011a}
\BIBentryALTinterwordspacing
M.~Berta, M.~Christandl, and R.~Renner, ``{The Quantum Reverse Shannon Theorem
  Based on One-Shot Information Theory},'' \emph{Communications in Mathematical
  Physics}, vol. 306, no.~3, pp. 579--615, September 2011. 
\BIBentrySTDinterwordspacing

\bibitem{Fawzi2017a}
\BIBentryALTinterwordspacing
H.~Fawzi and O.~Fawzi, ``{Efficient optimization of the quantum relative
  entropy},'' \emph{Journal of Physics A: Mathematical and Theoretical},
  vol.~51, no.~15, p. 154003, April 2018. 
\BIBentrySTDinterwordspacing

\bibitem{Datta2009}
\BIBentryALTinterwordspacing
N.~Datta, ``{Min- and Max-Relative Entropies and a New Entanglement
  Monotone},'' \emph{IEEE Transactions on Information Theory}, vol.~55, no.~6,
  pp. 2816--2826, June 2009. 
\BIBentrySTDinterwordspacing

\bibitem{Christandl2016}
\BIBentryALTinterwordspacing
M.~Christandl and A.~M{\"{u}}ller-Hermes, ``{Relative Entropy Bounds on
  Quantum, Private and Repeater Capacities},'' \emph{Communications in
  Mathematical Physics}, vol. 353, no.~2, pp. 821--852, Julyy 2017. 
\BIBentrySTDinterwordspacing

\bibitem{Berta2017a}
\BIBentryALTinterwordspacing
M.~Berta and M.~M. Wilde, ``{Amortization does not enhance the max-Rains
  information of a quantum channel},'' \emph{New Journal of Physics}, vol.~20,
  no.~5, p. 053044, May 2018. 
\BIBentrySTDinterwordspacing

\end{thebibliography}

\end{document}